\documentclass[reprint,amsmath,amssymb, aps,pra]{revtex4-1}
\usepackage{braket}
\DeclareMathOperator{\Tr}{Tr}
\usepackage{bm}
\usepackage{hyperref}
\usepackage{amsthm}
\newtheorem{Theorem}{Theorem}
\newtheorem*{Corollary}{Corollary}
\newtheorem*{Remark}{Remark}

\begin{document}

\title{Correlations and projective measurements in maximally entangled multipartite states}

\author{Arthur Vesperini}
\email[Email: ]{arthur_vespe@hotmail.fr}
\affiliation{DSFTA, University of Siena,\\ Via Roma 56, 53100 Siena, Italy
}%
\affiliation{Aix Marseille Université, Université de Toulon, CNRS, CPT, 13288 Marseille, France}
\date{\today}

\begin{abstract}
Multipartite quantum states constitute the key resource for quantum computation. The understanding of their internal structure is thus of great importance in the field of quantum information.
This paper aims at examining the structure of multipartite maximally entangled pure states, using tools with a simple and intuitive physical meaning, namely, projective measurements and correlations. 
We first show how, in such states, a very simple relation arises between post-measurement expectation values and pre-measurement correlations. 
We then infer the consequences of this relation on the structure of the recently introduced \textit{entanglement metric}, allowing us to provide an upper bound for the \textit{persistency of entanglement}.
The dependence of these features on the chosen measurement axis is underlined, and two simple optimization procedures are proposed, to find those maximizing the correlations.
Finally, we apply our procedures onto some prototypical examples.
\end{abstract}

\maketitle

\section{Introduction}

Entanglement, in addition to be one of the most historically puzzling properties of quantum mechanics, constitutes the main resource for quantum cryptography and computation, and quantum-based technologies. The quantum information community developed, in the past decades, an extensive number of approaches to characterize its abundant phenomenology \cite{GUHNE,horodecki}.

Full characterization of entanglement in multipartite states is a notoriously complex task \cite{ghz,Opt_estimation_ent,PhysRevA.95.062116,PhysRevA.67.022320,horodecki}. The mere definition of a measure of the entanglement of such states, is in itself a great challenge.

Yet, their complete analysis convokes numerous additional notions, as their \textit{$k$-separability}, \textit{$k$-producibility}, \textit{entanglement depth} or \textit{persistency of entanglement} \cite{GUHNE,BRS_persistent,horodecki}. 

Measurement processes and their understanding are of a major importance in the field of quantum computing. Measurement based quantum computation, which stands as a universal model of quantum computation, obviously heavily lie on the control of their effects on quantum states. \cite{Nielsen_2003,oneway_BR,meas_based_cluster,NIELSEN2006147}
\\

For the sake of simplicity, the present work will be focused on states $\ket{s}\in\mathcal{H}$, with $\mathcal{H}=\bigotimes_{\mu\in Q} \mathcal{H}_\mu$ the Hilbert space of dimension $2^N$ describing a set $Q$ of $N$ qubits. It should be noted that our results are generalizable to hybrid qudits systems with relatively few adjustments. 

We will hereafter denote $\sigma_k^\mu$, with $k=1,2,3$ the Pauli matrices acting on qubit $\mu$, $\bm{\sigma}^\mu=\big(\sigma_1^\mu,\sigma_2^\mu,\sigma_3^\mu\big)$ the Pauli vector acting on $\mu$, and $\sigma_{\bm{v}}^\mu=\bm{v}^\mu \cdot\bm{\sigma}^\mu = \sum\limits_{k=1,2,3}v_k^\mu\sigma_k^\mu$
 the Pauli observable on $\mu$ oriented in the direction $\bm{v}^\mu$. All of the vectors introduced in the following belong to $\mathbb{R}^3$. \\
 
As stated before, the ways of quantifying entanglement in multipartite states are manifold. In this work, we will solely refer to \textit{qubit-wise entanglement}, that is entanglement of bipartitions $(\mu,\mu^C)$, with $\mu\in Q$, and $\mu^C$ its complement on the set $Q$.

The reduced state of any qubit $\mu$ can be represented as a vector $\bra{s}\bm{\sigma}^\mu\ket{s}$ in the Bloch ball, i.e. a vector with norm lesser or equal to $1$. 
For a maximally entangled qubit $\mu$, it is the null vector, and we have
\begin{equation}\label{vanishing_expect_max_ent_mu}
\forall\, \bm{v}^\mu,\;  \bra{s}\sigma_{\bm{v}}^\mu\ket{s}=0,
\end{equation}
while fully factorizable qubits can always be represented by a well-defined vector on the Bloch sphere (i.e. of unit norm). This is consistent with the fact that a single qubit reduced density matrix $\rho^\mu=\Tr_{\mu^C}\big(\ket{s}\bra{s}\big)$ is maximally mixed, and has a maximal Von Neumann entropy $S(\rho_\mu)=-\Tr\big[\rho_\mu\log(\rho_\mu)\big]=\log(2)$. Yet the latter quantity is in fact the entropy of entanglement of the bipartition $(\mu,\mu^C)$, considered as a paradigmatic measure of bipartite entanglement \cite{Ent_entropy}. \\
 
The entanglement distance (ED), first defined in Ref. \cite{ED_2020}, is an entanglement measure for general multipartite pure states; it has been adapted in Ref. \cite{vesperini_entanglement_2023} to the more general framework of multipartite mixed states, and it has already found since then some interesting applications \cite{vafafard_nourmandipour,nourmandipour_entanglement_2021}. It finds its theoretical grounds on the Fubini-Study metric associated to the local-unitary invariant projective Hilbert space, called in this context the Entanglement Metric (EM). 

An EM of a state $\ket{s}$ is the real symmetric matrix of elements
\begin{equation}\label{ent_metric_elements}
    g_{\nu\mu}\big(\ket{s},\bm{v}^\mu,\bm{v}^\nu\big)= \bra{s} \sigma_{\bm{v}}^\mu \sigma_{\bm{v}}^\nu\ket{s} - \bra{s} \sigma_{\bm{v}}^\mu\ket{s}\bra{s} \sigma_{\bm{v}}^\nu\ket{s}.
\end{equation}
$g$ clearly happens to be a covariance matrix. This does not come as a surprise, as the indetermination in pure states, hence the (co-)variance, is always quantum in nature, and intimately linked to entanglement.

In particular, the diagonal elements of the EM, endowed with a minimization procedure, arise as an efficient measure of entanglement, with a behaviour similar to that of the entropy of entanglement. This is due to the fact that \eqref{vanishing_expect_max_ent_mu} implies that a maximally entangled qubit has variance $1$. The single-qubit ED is defined as
\begin{equation}\label{single-qubit_ED}
\begin{split}
E_\mu(\ket{s}):&=\min_{\bm{v}^\mu}\;g_{\mu\mu}\big(\ket{s},\bm{v}^\mu\big)\\
&= 1 - \max_{\bm{v}^\mu}|\bra{s}\sigma_{\bm{v}}^\mu\ket{s}|^2\\
&= 1 - |\bra{s}\bm{\sigma}^\mu\ket{s}|^2,
\end{split}
\end{equation}
which equates $1$ if $\mu$ is maximally entangled with the rest of the system, and $0$ if it is fully factorizable.

We choose here to use the latter definition of entanglement, which possesses the advantage of being very easy to compute, relative to the Von Neumann entropy. We further define the total entanglement of a state as $\sum\limits_{\mu\in Q}E_\mu(\ket{s})$.\\

In the present work, we primarily focus on maximally entangled state in the sense of Eq.\eqref{single-qubit_ED}. 
We start, in section \ref{sec:Theorems}, demonstrating a few simple theorems, which highlight the strong relationship between pre-measurement correlations and post-measurement average values, and show how the structure of the EM provides an upper bound to the persistency of entanglement. 
In section \ref{sec:optimization}, we derive two procedures in order to determine sets of measurement axis optimal with respect to pair-wise correlators or with respect to the induced total entanglement breaking; the question of the equivalency of these two optimal sets arises as an interesting open problem, to our best knowledge, yet to be solved. 
We apply our methods to a few examples in section \ref{sec:examples}.
Finally, in section \ref{sec:discussion}, we synthesize our results, and make a few remarks on possible continuations of this work; in particular, we examine the effect of several successive projective measurement on expectation values, and argue that a more thorough study of entanglement breaking in quantum states should investigate the behaviour of higher order correlations.

\section{First order projective measurements}\label{sec:Theorems}

Departing from a generic multipartite quantum state $\ket{s}\in\mathcal{H}$, the state $\ket{s'}$ obtained after a projective measurement of the qubit $\nu$ in the direction $\bm{m}^\nu$ write
\begin{equation}\label{proj_meas_sprime}
    \ket{s}\longrightarrow \ket{s'}=\frac{P_{\bm{m}}^\nu\ket{s}}{\sqrt{\bra{s}P_{\bm{m}}^\nu\ket{s}}},
\end{equation}
with
\begin{equation}
    P_{\bm{m}}^\nu:=\frac{\mathbb{I}+\sigma_{\bm{m}}^\nu}{2}
\end{equation}
the single qubit projector onto the eigenstate of $\sigma_{\bm{m}}^\nu$ of eigenvalue $1$.

The expectation value of an arbitrary qubit $\mu$ in the direction $\bm{v}^\mu$ in such a post-measurement state can be expressed as a function of the expectation values and the correlator of $\sigma_{\bm{m}}^\nu$ and $\sigma_{\bm{v}}^\nu$ in the initial state.

\begin{Theorem}\label{thm_1}
If $\ket{s}$ is maximally entangled in $\nu$ and $\mu$, and $\ket{s'}$ is the post-measurement state after a projective measure of $\sigma_{\bm{m}}^\nu$, then we have, for any measurement axis $\bm{v}^\mu$
\begin{equation}
\bra{s'}\sigma_{\bm{v}}^\mu\ket{s'}=\bra{s}\sigma_{\bm{v}}^\mu \sigma_{\bm{m}}^\nu\ket{s}
\end{equation}
\end{Theorem}

\begin{proof}
From Eq.\eqref{proj_meas_sprime} we draw
\begin{equation} 
\begin{split}
        \bra{s'}\sigma_{\bm{v}}^\mu\ket{s'}& = \frac{\bra{s} P_{\bm{m}}^\nu\sigma_{\bm{v}}^\mu P_{\bm{m}}^\nu\ket{s}}{\bra{s}P_{\bm{m}}^\nu\ket{s}}= \frac{\bra{s} P_{\bm{m}}^\nu\sigma_{\bm{v}}^\mu\ket{s}}{\bra{s}P_{\bm{m}}^\nu\ket{s}} \\
        &=\frac{\bra{s}\sigma_{\bm{v}}^\mu\ket{s} + \bra{s}\sigma_{\bm{v}}^\mu \sigma_{\bm{m}}^\nu\ket{s}}{1+\bra{s}\sigma_{\bm{m}}^\nu\ket{s}},
\end{split}
\end{equation}
where we used the fact that $P_{\bm{m}}^\nu$ and $\sigma_{\bm{v}}^\mu$ commute with each other, and that $P_{\bm{m}}^\nu$ is idempotent.
By hypothesis, Eq.\eqref{vanishing_expect_max_ent_mu} applies here, hence our Theorem.
\end{proof}

\begin{Corollary}
If $\ket{s}$ is maximally entangled in $\nu$ and $\mu$ and $\exists\, \bm{v}^\mu$ such that $|\bra{s}\sigma_{\bm{v}}^\mu \sigma_{\bm{m}}^\nu\ket{s}|=1$, then the measurement of $\nu$ along the axis $\bm{m}^\nu$ completely breaks the entanglement of $\mu$, i.e. $ E_\mu\big(\ket{s'}\big) =0$.
\end{Corollary}
\begin{proof}
If Theorem \ref{thm_1} applies, we can rewrite the post-measurement ED of $\mu$
\begin{equation}\label{Ent_postmeas}
    E_\mu\big(\ket{s'}\big) = 1 - \max_{\bm{v}^\mu}|\bra{s}\sigma_{\bm{v}}^\mu \sigma_{\bm{m}}^\nu\ket{s}|^2 ,
\end{equation}
which equates $0$ if the above condition is fulfilled.
\end{proof}

\begin{Theorem}\label{thm-2}
For any state $\ket{s}$, $\forall \mu,\nu\in Q$ such that $\bra{s}\sigma_{\bm{v}}^\mu \sigma_{\bm{v}}^\nu\ket{s}=1$, the operators $\sigma_{\bm{v}}^\mu$ and $\sigma_{\bm{v}}^\nu$ are equivalent with respect to $\ket{s}$ (they act on it in the same fashion). In particular, this implies, $\forall\eta$,
\begin{align}
& \bra{s}\sigma_{\bm{v}}^\eta \sigma_{\bm{v}}^\nu\ket{s} = \bra{s}\sigma_{\bm{v}}^\eta \sigma_{\bm{v}}^\mu\ket{s}, \label{correlator_equality}\\
\text{and }&\bra{s}\sigma_{\bm{v}}^\eta \sigma_{\bm{v}}^\mu \sigma_{\bm{v}}^\nu\ket{s}=\bra{s}\sigma_{\bm{v}}^\eta\ket{s}.
\end{align}
It also results that the measure of $\sigma_{\bm{v}}^\nu$ yields the implicit measure of $\sigma_{\bm{v}}^\mu$. 
\end{Theorem}
\begin{proof}
Under the hypothesis, $\ket{s}$ is eigenvector of $\sigma_{\bm{v}}^\mu \sigma_{\bm{v}}^\nu$ of eigenvalue $1$, and the following holds
\begin{equation}\label{sigmu-signu_equivalency}
\begin{split}
\sigma_{\bm{v}}^\mu \sigma_{\bm{v}}^\nu\ket{s}&  = \ket{s} \\
 \sigma_{\bm{v}}^\nu\ket{s}&  = \sigma_{\bm{v}}^\mu\ket{s} \\
 P_{\bm{v}}^\nu\ket{s}&  = P_{\bm{v}}^\mu\ket{s} \\
 P_{\bm{v}}^\nu\ket{s}&  = P_{\bm{v}}^\nu P_{\bm{v}}^\mu\ket{s}.
\end{split}
\end{equation}
\end{proof}

Let us now examine the consequences of these results on $g$. For the sake of clarity, we will hereafter drop its dependences and adopt the shortened notation $g_{\nu\mu}\big(\ket{s},\bm{v}^\mu,\bm{v}^\nu\big)=g_{\nu\mu}$.\\

\begin{Theorem}\label{thm-3}
Let $\ket{s}$ be a state maximally entangled $\forall\mu\in Q$ and $\{\bm{v}^\mu\}_\mu$ a choice of measurement directions such that $\forall\mu,\nu\in Q,\;g_{\mu\nu}=0$ or $\pm1$. Then:
    \begin{itemize}
        \item Up to a reordering of its indices (equivalently, a relabelling of the qubits), $g$ is diagonal by blocks filled with $\pm1$.
        \item The number $n$ of such blocks provides an upper bound to the persistency of entanglement $P_e$ of the state $\ket{s}$, i.e. $P_e\leq n$. In other words, the minimal number of local measurements necessary to fully break its entanglement is $n$ or less.
    \end{itemize}
\end{Theorem}
\begin{proof}
From Eq.\eqref{vanishing_expect_max_ent_mu}, it is clear that, $\forall\mu,\nu\in Q$, Eq.\eqref{ent_metric_elements} simplifies as
\begin{equation}\label{gmunu_maxent}
        g_{\nu\mu}= \bra{s} \sigma_{\bm{v}}^\mu \sigma_{\bm{v}}^\nu\ket{s},
\end{equation}
and, from Theorem \ref{thm-2}, we see that, $\forall\eta\in Q$, the following transitivity relation holds
\begin{equation}
|g_{\nu\mu}|=|g_{\nu\eta}|=1 \implies |g_{\mu\eta}|=1
\end{equation}

Added with the symmetry of $g$, this proves the first part of the Theorem. 

Using Theorem \ref{thm_1} together with Eq.\eqref{gmunu_maxent}, we can rewrite the diagonal elements of $g$ after a measure of $\sigma_{\bm{m}}^\nu$ as
\begin{equation}\label{gmumu_postmeas}
    g_{\mu\mu}' = 1 - |g_{\nu\mu}|^2 
\end{equation}
where $g'=g(\ket{s'})$. This straightforwardly implies that the measure of any qubit belonging to a block will collapse entanglement for the whole block, hence the second part of the Theorem.
\end{proof} 

The structure of $g$ gives valuable insights on \textit{persistency of entanglement} $P_e$ in such maximally entangled states. First introduced in \cite{BRS_persistent}, it quantifies the minimal number of measurements needed to completely disentangle a quantum state.

The notion of \textit{persistency of entanglement} is distinct from that of \textit{$k$-separability}\footnote{As a counter example, the state $\frac{1}{2}\big(\ket{0000}+\ket{0011}+\ket{1100}-\ket{1111}$ has an EM with $n=2$ blocks when $\forall\mu,\;\bm{v}^\mu=(0,0,1)$, hence has $P_e=2$ and yet is not biseparable.}. They are however related, as the entanglement of a maximally entangled $k$-separable state can clearly be fully broken by $k$ measurements. In other words, the persistency of entanglement provides an upper bound to the separability, and we have $n\geq P_e\geq k$.

Because $g$ only encodes informations on the effects of first order projective measurement (i.e. contains only two-points correlators), it may not capture the actual $P_e$, hence only provides an upper bound. This is due to the fact that new non-vanishing correlation patterns may arise after one or more projective measurements, entailing diminished $P_e$ relative to our first guess. The Briegel Raussendorf state with $N>4$, detailed in section \ref{subsec:BRS}, constitutes an example of such a situation. \\

Let us stress the dependence of $g$ on the set of unit vectors $\{\bm{v}^\mu\}$, representing directions of measurements. This point is of great importance because, as a different $g$ will arise from a different set $\{\bm{v}^\mu\}$, such are also the correlation patterns and subsequent entanglement breakings highlighted in the above. In order to bound $P_e$ as finely as possible, one thus needs to find the appropriate set $\{\bm{v}^\mu\}$.

\section{Optimization of the measure}\label{sec:optimization}
 
We are now looking for the measurement directions $\bm{v}^\nu$ optimizing the correlators in Eq.\eqref{ent_metric_elements}.\\ 

First, let us perform the pair-wise optimization of the correlators
\begin{equation}\label{optimization_pairwise_corr}
\bra{s}\sigma_{\bm{v}}^\mu\sigma_{\bm{v}}^\nu\ket{s} = \sum_{i,j=1}^3 v_i^\mu v_j^\nu \bra{s}\sigma_i^\mu\sigma_j^\nu\ket{s} = 
(\bm{v}^\mu)^T\bm{C}_s^{\mu\nu}\bm{v}^\nu,
\end{equation}
where $\bm{C}_s^{\mu\nu}$ is the non-symmetric real matrix of elements
\begin{equation}
\big(\bm{C}_s^{\mu\nu}\big)_{ij}=\bra{s}\sigma_i^\mu\sigma_j^\nu\ket{s}
\end{equation}
i.e. the spin correlation matrix. Optimization with respect to both measurement axis yields
\begin{equation}
\begin{split}
\bm{C}_s^{\mu\nu}\bm{v}^\nu&=\lambda\bm{v}^\mu\\
(\bm{v}^\mu)^T\bm{C}_s^{\mu\nu}&=\lambda(\bm{v}^\nu)^T,
\end{split}
\end{equation}
where the superscript $T$ stands for transpose. One can check by easy direct calculation that the eigenvalues indeed coincide. By insertion of the transpose of these two equations, we infer the eigenvalue equations
\begin{equation}
\begin{split}
(\bm{C}_s^{\mu\nu})^T\bm{C}_s^{\mu\nu}\bm{v}^\nu&=\lambda^2\bm{v}^\nu\\
\bm{C}_s^{\mu\nu}(\bm{C}_s^{\mu\nu})^T\bm{v}^\mu&=\lambda^2\bm{v}^\mu,
\end{split}
\end{equation}
of which largest eigenvalue solutions evidently correspond to the measurement directions maximizing Eq.\eqref{optimization_pairwise_corr}. \\

We can instead be interested in finding the measurement axis $\bm{m}^\nu$ that optimally disentangle the entire state $\ket{s}$. Provided the latter is maximally entangled, we have, as a consequence of Theorem \ref{thm_1}
\begin{equation}
\big|\bra{s'}\bm{\sigma}^\mu\ket{s'}\big|^2 = \sum_{i=1}^3\big|\bra{s'}\sigma_i^\mu\ket{s'}\big|^2 = \sum_{i=1}^3\big|\bra{s}\sigma_i^\mu\sigma_{\bm{m}}^\nu\ket{s}\big|^2,
\end{equation}
hence the total entanglement after the measure
\begin{equation}\label{optimization_totent}
\begin{split}
    E\big(\ket{s'}\big)&=\sum_{\mu} E_\mu(\ket{s'})= N - 1  - \sum_{\mu\in \nu^C}|\bra{s'}\bm{\sigma}^\mu\ket{s'}|^2\\
    &=N - 1 - \sum_{\mu\in \nu^C}\sum_{i=1}^3\big|\bra{s}\sigma_i^\mu\sigma_{\bm{m}}^\nu\ket{s}\big|^2\\
    &=N - 1 - \sum_{\mu\in \nu^C}\sum_{i=1}^3\bra{s}\sigma_{\bm{m}}^\nu\sigma_i^\mu\ket{s}\bra{s}\sigma_i^\mu\sigma_{\bm{m}}^\nu\ket{s}\\
    &=N - 1 - \bra{s}\sigma_{\bm{m}}^\nu\Big(\sum_{\mu\in \nu^C}\sum_{i=1}^3\sigma_i^\mu\ket{s}\bra{s}\sigma_i^\mu\Big)\sigma_{\bm{m}}^\nu\ket{s}\\
    &=N - 1 - \sum_{j,k=1}^3 m_j^\nu m_k^\nu\bra{s}\sigma_j^\nu\Big(\sum_{\mu\in \nu^C}\sum_{i=1}^3\sigma_i^\mu\ket{s}\bra{s}\sigma_i^\mu\Big)\sigma_k^\nu\ket{s},
\end{split}
\end{equation}
with $\nu^C=Q\setminus\{\nu\}$.
Let us define $\bm{B}_s^{\nu}$, the measurement-induced entanglement breaking matrices (MIEB) of the state $\ket{s}$, of elements
\begin{equation}\label{MIEB-matrix}
    \Big(\bm{B}_s^{\nu}\Big)_{jk}=\bra{s}\sigma_j^\nu\,\Sigma_s^{\nu^C}\,\sigma_k^\nu\ket{s},
\end{equation}
with $\Sigma_s^{\nu^C}=\sum\limits_{\mu\in \nu^C}\sum\limits_{i=1}^3\sigma_i^\mu\ket{s}\bra{s}\sigma_i^\mu$.
Its diagonalization straightforwardly yields the results of the optimization problem, as the eigenvectors $\widetilde{\bm{m}}^\nu$ associated to the largest (resp. smallest) eigenvalue of $\bm{B}^{\nu}\big(\ket{s}\big)$ corresponds to the minimum (resp. maximum) of Eq.\eqref{optimization_totent}. The eigenvalues themselves simply equate the total additional loss of entanglement after the corresponding measurement. It results that a comparison of the spectra of the $N$ MIEB matrices  allow to easily find the ``weak spots'' of $\ket{s}$, that is the qubits of which the measurement can break entanglement the most.

Interestingly, the largest eigenvalue of $\bm{B}_s^{\nu}$ might have multiplicity greater than one. In such cases, any of the corresponding eigenvectors or linear combination of them maximize Eq.\eqref{optimization_totent}.

This result can straightforwardly be adapted to find the measurement axis optimizing the entanglement breaking of a subset $Q'\subset Q$ with $\nu\notin Q'$, by simply replacing $\Sigma_s^{\nu^C}$ with $\Sigma_s^{Q'}=\sum\limits_{\mu\in Q'}\sum\limits_{i=1}^3\sigma_i^\mu\ket{s}\bra{s}\sigma_i^\mu$.\\

The question naturally arises as to know whether the solutions of Eq.\eqref{optimization_pairwise_corr} correspond in general to those of Eq.\eqref{optimization_totent}.

\begin{Remark}
The existence, in the general case, of one (or several) set of measurement axis $\{\bm{v}^\nu\}_\nu^{opt}$ optimizing \textit{simultaneously} every correlator in the EM remains unclear. \footnote{To clarify the meaning of this problem, let us imagine a state $\ket{s}$ such that $\nexists\{\bm{v}^\nu\}_\nu^{opt}$. Then there exists some qubit $\mu$ of which the correlation with $\nu$ is maximal in a direction $v_1^\mu$, and the correlation with $\eta\neq\nu$ is maximal in a direction $v_2^\mu\neq v_1^\mu$.}
\begin{itemize}
\item If $\{\bm{v}^\nu\}_\nu^{opt}$ indeed always exists, it would be enough, to completely probe the pair-wise correlation patterns of a given state, to solve $\lceil N/2 \rceil$ equations of the form of Eq.\eqref{optimization_pairwise_corr} or, equivalently, to diagonalize $N$ MIEB \eqref{MIEB-matrix}.

\item If, on the contrary, it doesn't, such a complete probing requires to solve all of the $N(N-1)/2$ equations of the form of Eq.\eqref{optimization_pairwise_corr}.
\end{itemize}
\end{Remark}

Yet, in all of the example we considered, the solutions of Eq.\eqref{optimization_pairwise_corr} equate those of Eq.\eqref{optimization_totent}, hence we were able to determine $\{\bm{v}^\nu\}_\nu^{opt}$. We were not able find any counter-example. We will thus assume from now on that this set exists, in particular in Section \ref{sec:examples}, as this assumption allows us to remain much more concise, by only diagonalizing the $N$ MIEB.

\section{Examples and applications}\label{sec:examples}

\subsection{Briegel-Raussendorf states}\label{subsec:BRS}

The Briegel-Raussendorf states (BRS) form a family of quantum states, introduced in \cite{BRS_persistent}, They are, up to local unitary transformation, equivalent to cluster states (also coined as graph states in the litterature), which were proposed as a model for the measurement-based quantum computer \cite{oneway_BR,NIELSEN2006147,Nielsen_2003,meas_based_cluster}.\\

BRS are defined, for an arbitrary number of qubits $N$ on a $d$-dimensional lattice, as
\begin{equation}
\ket{\phi(\varphi)} = U(\varphi)\ket{+}^{\otimes N},
\end{equation}
where $\ket{+}^\mu$ is the eigenstate of the operator $\sigma_1^\mu$ of eigenvalue $1$, and
\begin{equation}\label{UBRS}
U(\varphi) = \exp\Big\{-i\varphi\sum_{<\mu,\nu>}P_0^\mu P_1^\nu\Big\},
\end{equation}
where the summation runs over all the pairs of nearest neighbours, and $P_{0(1)}^\mu$ is the projector onto the eigenstates of $\sigma_3^\mu$ of eigenvalue $\pm 1$.
We are solely interested in the case $\varphi=\pi$, as it results in a maximally entangled state. We will furthermore restrict ourselves to the study of the $1$-dimensional case.
\begin{equation}
\ket{\phi_N} = \prod_{\mu=0}^{N-2}\frac{1}{2}\big(\mathbb{I}-\sigma_3^\mu+\sigma_3^{\mu+1}+\sigma_3^\mu\sigma_3^{\mu+1}\big)\ket{+}^{\otimes N}.
\end{equation}

It has been shown in \cite{ED_2020} that the $E_\mu(\ket{\phi_N})=0$, $\forall\mu$. \\

\paragraph*{$N=3$.}

It is known that the $3$-qubits BR state is local unitary (LU) equivalent to the Greenberger-Horne-Zeilinger state \cite{BRS_persistent}, a prototypical example of maximally entangled and maximally correlated state. We hence expect that $\exists\{\bm{v}^\mu\}$ such that $\forall\mu,\nu,\,g_{\mu\nu}=1$.

We only need to compute the three following MIEB matrices \eqref{MIEB-matrix}
\begin{equation}
\begin{split}
\bm{B}_{\phi_3}^{0} & = \bm{B}_{\phi_3}^{2}= \begin{pmatrix}
2 & 0 & 0 \\
0 & 0 & 0 \\
0 & 0 & 0 
\end{pmatrix}\text{, hence  }
\bm{v}^0=\bm{v}^2=(1,0,0) \text{, and }\\
\bm{B}_{\phi_3}^{1} &=\begin{pmatrix}
0 & 0 & 0 \\
0 & 0 & 0 \\
0 & 0 & 2 
\end{pmatrix}\text{, hence  } \bm{v}^1=(0,0,1).
\end{split}
\end{equation}
The maximal eigenvalues of these matrices equate $2$ because each qubit is maximally correlated with two others.

This yields the EM
\begin{equation}
g\big(\ket{\phi_3},\{\bm{v}^\mu\}\big)=\begin{pmatrix}
1 & 1 & -1 \\
1 & 1 & -1 \\
-1 & -1 &1 .
\end{pmatrix}
\end{equation}
A single measurement is thus sufficient to completely disentangle $\ket{\phi_3}$. We have $P_e(\ket{\phi_3})=1$, hence this state is genuinely entangled.\\

\paragraph*{$N=4$.}

In this case, we have
\begin{equation}
\begin{split}
\bm{B}_{\phi_4}^{0} &=\bm{B}_{\phi_4}^{3} =\begin{pmatrix}
1 & 0 & 0 \\
0 & 0 & 0 \\
0 & 0 & 0 
\end{pmatrix}\text{, hence } \bm{v}^0=\bm{v}^3=(1,0,0), \text{ and}\\
\bm{B}_{\phi_4}^{1}& =\bm{B}_{\phi_4}^{2} =\begin{pmatrix}
0 & 0 & 0 \\
0 & 0 & 0 \\
0 & 0 & 1 
\end{pmatrix}\text{, hence } \bm{v}^1=\bm{v}^2=(0,0,1),
\end{split}
\end{equation}
thus the EM writes
\begin{equation}
g(\{\bm{v}^\mu\})=\begin{pmatrix}
1 & 1 & 0 & 0 \\
1 & 1 & 0 & 0 \\
0 & 0 & 1 & -1 \\
0 & 0 & -1 & 1  
\end{pmatrix},
\end{equation}
hence $P_e(\ket{\phi_4})\leq 2$. \\

\paragraph*{$N>4$.}

In general, for a chain of $N$ qubits, we have
\begin{equation}
\begin{split}
\bm{B}_{\phi_N}^{0} &=\bm{B}_{\phi_N}^{N)} =\begin{pmatrix}
1 & 0 & 0 \\
0 & 0 & 0 \\
0 & 0 & 0 
\end{pmatrix}, \text{ hence } \bm{v}^0=\bm{v}^{N}=(1,0,0), \text{ and } \\
\bm{B}_{\phi_N}^{1} &=\bm{B}_{\phi_N}^{N-1)} =\begin{pmatrix}
0 & 0 & 0 \\
0 & 0 & 0 \\
0 & 0 & 1 
\end{pmatrix}, \text{ hence } \bm{v}^1=\bm{v}^{N-1}=(0,0,1),
\end{split}
\end{equation}
while $\bm{B}_{\phi_N}^{\nu}=\bm{0},\;\forall\nu\notin \{0,1,N-1,N\}$. 

The EM hence contains two $2\times 2$ blocks in its upper left and lower right corners, while the remaining part is just a diagonal filled with ones (which can be seen as trivial blocks of size $1\times1$), hence the number of its blocks is $N-2$. Its persistency of entanglement is however known to be $P_e(\ket{\phi_N})=\lfloor \frac{N}{2} \rfloor$ \citep{BRS_persistent}.

This discrepancy is a paradigmatic example of the fact that Theorem \ref{thm-3} is in general insufficient to capture exactly $P_e$. Measurements can indeed be performed on $\ket{\phi_N}$, after which some elements of the EM will cease to be null. The post-measurement EM thus exhibits some new non trivial blocks, accounting for this discrepancy. \\

\subsection{Supersinglet states}

The supersinglet states, first introduced in \cite{Supersinglets}, form a class of maximally entangled pure states with the property of being invariant under any simultaneous LU operation acting on all qubits, i.e. $U^{\otimes N}\ket{S_N}=e^{i\phi}\ket{S_N}$, with $U$ an arbitrary LU operator. \cite{Supersinglets,Supersinglets2,GUHNE}. \\

From this sole property, a number of facts can be drawn.

First, $\forall U$, $\forall\bm{v}^\nu$, we have
\begin{equation}
\begin{split}
\bra{S_N}\sigma_{\bm{v}}^\nu\ket{S_N} & = \bra{S_N}U^{\dagger\otimes N}\sigma_{\bm{v}}^\nu U^{\otimes N}\ket{S_N} \\
&= \bra{S_N}U^{\nu\dagger}\sigma_{\bm{v}}^\nu U^\nu\ket{S_N} = \bra{S_N}\sigma_{\bm{v}'}^\nu\ket{S_N},
\end{split}
\end{equation}
yet Pauli expectation values cannot be isotropic unless null, hence Eq.\eqref{vanishing_expect_max_ent_mu} is verified and $E_\nu(\ket{S_N})=1,\,\forall\nu$.

Theorem \ref{thm_1} hence applies here and, choosing $U$ to be a rotation around the axis $\bm{m}^\nu$, hence leaving $\sigma_{\bm{m}}^\nu$ unchanged, we can write
\begin{equation}
\begin{split}
\bra{S_N'}\sigma_{\bm{v}}^\mu\ket{S_N'}&=\bra{S_N}\sigma_{\bm{v}}^\mu\sigma_{\bm{m}}^\nu\ket{S_N}\\
&=\bra{S_N}U^{\dagger\otimes N}\sigma_{\bm{v}}^\mu\sigma_{\bm{m}}^\nu U^{\otimes N}\ket{S_N}\\
&=\bra{S_N}\big(U^{\mu\dagger}\sigma_{\bm{v}}^\mu U^\mu\big)\sigma_{\bm{m}}^\nu \ket{S_N}\\
&=\bra{S_N}\sigma_{\bm{v}'}^\mu\sigma_{\bm{m}}^\nu\ket{S_N}=\bra{S_N'}\sigma_{\bm{v'}}^\mu\ket{S_N'},
\end{split}
\end{equation}
where $\ket{S_N'}$ is the state post-measurement of $\sigma_{\bm{m}}^\nu$.
The same argument as above leads to $\bra{S_N}\sigma_{\bm{v}}^\mu\sigma_{\bm{m}}^\nu\ket{S_N}\neq 0$ if and only if $\bm{v}^\mu=\bm{v}^{\mu\prime}$, that is if $\bm{v}^\mu=\pm\bm{m}^\nu$.

It results $\bm{B}_{S_N}^{\mu}\propto\mathbb{I},\;\forall\mu\in Q$. This means that, provided that the qubits are measured along the same axis, the correlators are always maximal. The optimal set of measurement axis is thus any uniform set $\{\bm{v}^\mu\}^{uni}$. \\

The whole class of supersinglet states of four qubits is spanned by the following two states 
\begin{equation}
\begin{split}
\ket{S^{1}_4}&=\frac{1}{\sqrt{3}}\Big(\ket{0011} + \ket{1100} -\frac{1}{2}\big( \ket{0101} + \ket{0110} + \ket{1001} + \ket{1010}\big)\big)\\
\ket{S^{2}_4}&=\frac{1}{2}\Big(\ket{0101} + \ket{1010} - \ket{0110} - \ket{1001} \big)\big).
\end{split}
\end{equation}
We can hence write a general 4-qubits supersinglet state as $\ket{S_4(a,b)}=a\ket{S^{1}_4}+b\ket{S^{2}_4}$, with $|a|^2+|b|^2$.

We can choose arbitrarily a single measurement axis, for instance $\bm{x}=(1,0,0)$, and straightforwardly compute the optimal EM
\begin{equation}
g(\{\bm{v}^\mu\}^{uni})=\begin{pmatrix}
1 & \alpha & \gamma & \beta \\
\alpha & 1 & \beta & \gamma \\
\gamma & \beta & 1 & \alpha \\
\beta & \gamma & \alpha & 1  
\end{pmatrix},
\end{equation}
with 
\begin{equation}
\begin{split}
\alpha&=\frac{|a|^2}{3} - |b|^2,\\
\beta&=\frac{2}{3}(\sqrt{3}\Re(\bar{a}b)-|a|^2), \text{and}\\
\gamma&=-\frac{2}{3}\big(\sqrt{3}\Re(\bar{a}b)+|a|^2\big),
\end{split}
\end{equation}
where the bar over a letter denotes complex conjugation and $\Re(\bar{a}b)$ is the real part of the complex number $\bar{a}b$.\\

This simple expression for the EM allow us to remark some interesting cases arising for specific values of the parameters $a$ and $b$.

Trivially, if $a=0$ and $|b|=1$, we immediately get a block diagonal matrix containing two blocks, as is to be expected, since $\ket{S_4^{(2)}}$ is a tensor product of two Bell states $\ket{\phi_-}=\frac{1}{\sqrt{2}}(\ket{01}-\ket{10})$. In this case, $P_e(\ket{S_4(0,1)})=2$ and the state is of course biseparable.

If $a=\frac{\sqrt{3}e^{i\phi_a}}{2}$ and $b=\frac{e^{i\phi_b}}{2}$, a few calculations lead to 
\begin{equation}
g(\{\bm{v}^\mu\}^{uni})=\begin{pmatrix}
1 & 0 & -c^2 & -s^2 \\
0 & 1 & -s^2 & -c^2 \\
-c^2 & -s^2 & 1 & 0 \\
-s^2 & -c^2 & 0 & 1  
\end{pmatrix},
\end{equation}
with $c=\cos(\frac{\phi_a-\phi_b}{2})$ and $s=\sin(\frac{\phi_a-\phi_b}{2})$.
Thus, for $\phi_a=\phi_b$ and for $\phi_b=\phi_a+\pi$, the EM of this state, up to qubits permutations, contains two blocks, and $P_e(\ket{S_4(\frac{\sqrt{3}e^{i\phi_a}}{2},\frac{e^{i\phi_a}}{2})})=P_e(\ket{S_4(\frac{\sqrt{3}e^{i\phi_a}}{2},\frac{e^{i(\phi_a+\pi)}}{2})})=2$.

\section{Discussion}\label{sec:discussion}

In this work, we showed how, in pure quantum states, entanglement entails a strong link between correlations and post-measurement expectation values.

The EMs, i.e. covariance matrices, contain valuable informations on the statistics of post-measurement states, and on the patterns that can emerge from projective measurements of entangled states. In particular, its block structure is directly linked to the persistency of entanglement. 

We further provided two straightforward procedures of optimization of the Pauli correlators, and observe that they might not, in principle, yield equivalent results. By doing so, we unravel an opened problem, which hasn't, to our best knowledge, been tackled with yet, namely the existence of a set of measurement axis simultaneously optimizing all of these correlators.

These procedures further allow us to recover, if it exists, the optimal EM, along with an upper bound for the persistency of entanglement. \\

Unfortunately, as emphasized above, the information retrieved by the use of the EM is incomplete: since the effect of more than one measurement are not accounted for in this framework, our approach fails to recover a number of important features of some complex entangled states (as BR states or cluster states), amongst which the exact persistency of entanglement.

Multipartite maximally entangled states might indeed possess qubits with only vanishing two-points correlations, regardless of the choice of measurement axis. Somehow counter-intuitively, the measurement of such qubits, though disentangling the concerned qubit from the rest of the system, do not break any entanglement on the latter. 
Yet it modifies the state, and in particular might bring along some new non-vanishing correlators (i.e. off-diagonal terms in the EM).

This observation motivates the study of higher order measurement schemes. 

Let us consider an ordered subset $\mathcal{M}\subset Q$ of $M$ qubits on which successive projective measurements are performed. 

The generalization of \eqref{proj_meas_sprime} then yields
\begin{equation}
    \ket{s}\longrightarrow \ket{s^{_\mathcal{M}}}=\frac{\prod\limits_{\nu\in\mathcal{M}} P_{\bm{m}}^\nu\ket{s}}{\sqrt{\bra{s}\prod\limits_{\nu\in\mathcal{M}}P_{\bm{m}}^\nu\ket{s}}},
\end{equation}
and, accordingly, the expectation value of an arbitrary unmeasured qubit $\mu$ after such a series of measurement
\begin{equation}
\setlength{\jot}{10pt}
    \bra{s^{_\mathcal{M}}}\sigma_{\bm{v}}^\mu\ket{s^{_\mathcal{M}}} = \frac{\sum\limits_{k=0}^{M-1}\sum\limits_{X\in[\mathcal{M}]^k}\bra{s}\sigma_{\bm{v}}^\mu \prod\limits_{\nu\in X}\sigma_{\bm{m}}^{\nu}\ket{s}}{\sum\limits_{k=0}^{M-1}\sum\limits_{X\in[\mathcal{M}]^k}\bra{s} \prod\limits_{\nu\in X}\sigma_{\bm{m}}^{\nu}\ket{s}},
\end{equation}
where $[\mathcal{M}]^k$ is the set of all the unordered $k$-subsets (i.e. subsets of cardinal $k$) of $\mathcal{M}$.\\

It results that, if all the two-points correlators vanish, one can examine higher order correlators to investigate the breaking of entanglement after a series of measurement rather a single one. 

It would be of great interest to pursue this research by a thorough study of higher order covariance tensors, or to devise new procedures and methods, in order to grasp the effects of series of measurements on entangled states, in a more exhaustive fashion.

\begin{acknowledgments}
The author acknowledges support from the RESEARCH SUPPORT
PLAN 2022 - Call for applications for funding allocation to research projects curiosity driven (F CUR) - Project ”Entanglement Protection of Qubits’ Dynamics in a Cavity” – EPQDC , and the support by the Italian National Group of
Mathematical Physics (GNFM-INdAM).
The author further acknowledges useful discussions with Roberto Franzosi.
\end{acknowledgments}

\bibliography{references}

\end{document}